\documentclass[prd,showpacs,preprintnumbers,groupedaddress]{revtex4-1}

\usepackage{amsthm}
\theoremstyle{plain}
\newtheorem{theorem}{Theorem}

\newtheorem{lemma}{Lemma}
\newtheorem*{corollary}{Corollary}

\usepackage{cases}
\usepackage{color}
\usepackage{comment}
\usepackage{ulem}

\newcommand{\llaabel}[1]{\label{#1}}

\begin{document}


\title{Rotating accretion flows in $D$ dimensions: 
sonic points, critical points and photon spheres}


\author{Yasutaka Koga}
\author{Tomohiro Harada}
\affiliation{Department of Physics, Rikkyo University, Toshima, Tokyo 171-8501, Japan}


\date{\today}

\begin{abstract}
We give the formulation and the general analysis of the rotational
 accretion problem on $D$-dimensional spherical spacetime and
 investigate sonic points and critical points.
First, we construct the simple two-dimensional rotating accretion flow
 model in general $D$-dimensional static spherically symmetric spacetime and formulate the problem.
The flow forms a two-dimensional disk lying on the equatorial plane and the disk is assumed to be geometrically thin and has uniform distribution in the polar angle directions.
Analyzing the critical point of the problem, we give the conditions for the critical point and its classification explicitly and show the coincidence with the sonic point for generic equation of state (EOS).
Next, adopting the EOS of ideal photon gas to the analysis, we reveal that there always exists a correspondence between the sonic points and the photon spheres of the spacetime.
Our main result is that the sonic point of the rotating accretion flow
 of ideal photon gas must be on (one of) the unstable photon sphere(s)
 of the spacetime
in arbitrary spacetime dimensions. 
This paper extends this correspondence for spherical flows shown in the authors' previous work to rotating accretion disks.
\end{abstract}


\maketitle


\tableofcontents

\section{Introduction}
\llaabel{sec:introduction}
The accretion of fluid onto objects is a basic
problem in astrophysics and
the most important issue concerning growth of stars and black holes.
In an observational view point, the accretion is considered to be responsible
for the X-ray emission due to the compression of the fluid.
This is also connected to the observations of strong gravity fields in a general relativistic context.
The accretion problem has been widely studied in the cases of Newtonian gravity~\cite{bondi}, Schwarzschild spacetime~\cite{michel}, Schwarzschild (anti-)de Sitter spacetime~\cite{mach} and generic spherically symmetric spacetime~\cite{chaverra}.
One of the interesting features is the existence of transonic flow and a sonic point (or critical point), that is, flow which experiences transition between subsonic and supersonic states and a point at which such a transition occurs.
This generally appears in accretion problems and plays key role in the analysis.
\par
A photon sphere is a sphere of spacetime on which motion of photons, or null geodesics, can take circular orbits.
In astrophysical cases, black holes usually have photon spheres near the horizon.
This structure is characteristic of strong gravitational fields and it is well known that Schwarzschild spacetime has a photon sphere on the areal radius $r=3M$ and the sphere is unstable, that is, the circular orbits on the sphere are unstable orbits.
The photon sphere has been widely studied in its various aspects;
For optical observations of black holes through background light emission, the photon sphere determines the size of the black hole shadow.
In the case of the Schwarzschild black hole for example, we can see their relation from the calculation  by Synge~\cite{synge};
Properties of gravitational waves from black holes are also closely related to the photon sphere.
It is known that the frequencies of quasinormal modes are determined by the parameters of null geodesic motions on and near the photon sphere in various situations~\cite{cardoso}~\cite{hod};
The nature of photon sphere itself and the generalization to non-spherically symmetric spacetimes are also investigated in~\cite{claudel} and the citation.
\par
In the study by Mach et al.~\cite{mach}, it was revealed that only for the case of the accretion of radiation fluid, the radius of the sonic point is $3M$.
This radius coincides with the photon sphere of the spacetime.
In the previous paper of the present authors~\cite{koga} 
(hereafter referred to as paper I), 
the analysis of the spherical accretion problem by Chaverra and
Sarbach~\cite{chaverra} is extended to general, static and spherically
symmetric spacetime in arbitrary dimensions and 
a correspondence is found in much more general situations:
On general static spherically symmetric spacetime, the sonic point of the spherical accretion flow of radiation fluid must be on the unstable photon sphere and this correspondence holds in arbitrary dimensions.
In an observational view point, this correspondence connects two independent observations,
the observation of lights from sources behind a black hole and the
observation of emission from accreted radiation fluid onto the hole,
because the size of the shadow of the hole is determined by the radius
of the photon sphere and the accreted fluid can signal the sonic point.
The stationary accretion flow 
of radiation fluid in spherically symmetric spacetime can model 
accretion onto primordial black holes in the 
radiation-dominated era of the Universe~(\cite{carr} and references therein) 
and accretion of cosmic
microwave background onto black holes in the recent phase of the 
Universe. The time reverse of this system can model the outflow of 
Hawking radiation from a black hole~\cite{carter}.
In a theoretical view point, 
the correspondence in very general stuations suggests that this is
of microphysical origin because the description of photon gas by 
radiation fluid is derived for the isotropic distribution of photons 
by purely local physics.
\par
In this paper, we show there exists the correspondence between sonic points and photon spheres in the case of rotational accretion of ideal photon gas.
In the first half, we construct our rotational accretion flow model which forms a disk on an equatorial plane of $D$-dimensional static spherical symmetric spacetime ($D\ge3$) and analyze its sonic points for general equation of state (EOS) of fluid.
The metric is given by 
\begin{equation}
\llaabel{eq:metric}
ds^2=-f(r)dt^2+g(r)dr^2+r^2d\Omega^2_{D-2}\qquad(D\ge3),
\end{equation}
where the condition $0<f,g<\infty$ is assumed and $d\Omega^2_{D-2}$ is the unit $(D-2)$-sphere metric given by
\begin{equation}
d\Omega_{D-2}^2=d\theta_1^2+\cdots+\sin^2\theta_1\cdots\sin^2\theta_{D-4}d\theta_{D-3}^2
+\sin^2\theta_1\cdots\sin^2\theta_{D-3}d\phi^2.
\end{equation}
The polar and the azimuthal angle coordinates are denoted by $\theta_i\left(i=1,...,D-3\right)$ and $\phi$, respectively.
In the latter half, applying the analysis to radiation fluid, we show there exists correspondence in the rotational accretion.
\par
We start with the construction of our accretion disk model and the formulation of the problem and present the definition of critical point and sonic point and their relation in Sec.~\ref{sec:construction}.
In Sec.~\ref{sec:criticalpoint}, we explicitly analyze the conditions for the critical point and its classification without specifying EOS.
In Sec.~\ref{sec:radiation}, Applying the EOS of the ideal photon gas (or, radiation fluid) to the analysis, we derive the conditions about the critical point and its classifications in that case.
Then, recalling the conditions for the photon spheres from paper I~\cite{koga}, we see that there also exists the correspondence between the sonic point and the photon sphere in the rotational accretion problem.
The conclusion is given in Sec.\ref{sec:conclusion}.

\section{Rotating accretion disk,  critical point and sonic point
\llaabel{sec:construction}
}
Here assuming three conservation laws and the metric (\ref{eq:metric}), we formulate the accretion problem.
The conditions for the critical point and the relation between the critical points and the sonic points are also given in the subsequent subsections.
\par
We assume the following three equations, the first law of thermodynamics, continuity equation and energy-momentum conservation with perfect fluid:
\begin{subnumcases}
{}
\llaabel{first_law}
dh=Tds+n^{-1}dp\\
\llaabel{num_conservation}
\nabla_\mu J^\mu=0\\
\llaabel{em_conservation}
\nabla_\mu T^\mu_\nu=0,
\end{subnumcases}
where $J^\mu := nu^\mu$ is the number current and $T^\mu_\nu=nhu^\mu u_\nu+p\delta^\mu_\nu$ is the energy-momentum tensor of the perfect fluid.
The quantities $h,T,s,n,p$ and $u^\mu$ represent the enthalpy per particle, the temperature,
the entropy per particle, the number density, the pressure and the
4-velocity of the fluid.

\subsection{Configuration of the accretion disk}
We assume several conditions for the accretion disk.
Although there are many other possibilities about disk thickness and the
vertical equilibrium, we follow the simplest model used by Abraham,
Bilic and Das\cite{abraham}. Although their model is based on the axial coordinate, $z$, in rotational
spacetime, however, we here choose the polar coordinate,
$\theta_{i}$, in spherical spacetime. 
In this sense, strictly speaking, our 
model is not the same as their model.
\par
We assume the following conditions for the accretion disk:
\begin{enumerate}
\item The disk lying on the equatorial plane has symmetries which respects the background geometry;
	\begin{itemize}
	\item Stationarity along the Killing vector $\xi_{(t)}=\partial_t$
	\item Rotational symmetry along the Killing vector $\xi_{(\phi)}=\partial_\phi$
	\item Reflection symmetry respective to the equatorial plane
	\end{itemize}
	\llaabel{item-cnd-sym}
\item The matter distribution is uniform in the $\theta_i$-direction: The number density $n$, the pressure $p$, the entropy and the components of the velocity $u^\mu$ are independent of $\theta_{i}$\llaabel{item-cnd-thetaindepend}
\item The disk is sufficiently thin so that 
we can ignore terms of the second or higher order of $\theta_i-\pi/2$ compared to that of the zero-th order.
\item The equilibrium between the inside and the outside of the disk
      surface in the polar direction is achieved
       by pressure of external rarefied gas.  
\end{enumerate}
The {\it equatorial plane} here means 
the 2-plane $(r,\phi)$ with all the polar angles $\theta_i=\pi/2\;
(i=1,2,...,D-3)$ and the {\it disk} consists of the 2-dimensional
 plane with the $(D-3)$-dimensional
(sufficiently small) thickness.
Note that the disk surface is a spatial ($D-2$)-space of $\theta_i=const.$ and $u^{\theta_i}=0$ everywhere in the disk as the consequence of conditions \ref{item-cnd-sym} and \ref{item-cnd-thetaindepend}.
\par
In accretion problems, we are interested in the accretion rate and it is also important for analysis of the dynamics.
Consider $t=const.$ hypersurface and the disk volume $\Sigma_D$ on
the hypersurface. Let $\Sigma_r$ be the region inside the radius $r$ in $\Sigma_D$.
The total particle number $N(t,r)$ in the volume $\Sigma_r$ at the time
$t$ is given by
\begin{equation}
N(t,r):=\int_{\Sigma_r}nu^\mu d\Sigma_\mu.
\end{equation}
Then the accretion rate $\dot{N}(t,r)$ of the particle number into the region $\Sigma_r$ is given by
\begin{equation}
\dot{N}(t,r)=-\int_{S_r}nu^r\sqrt{fg}r^{D-2}d\Omega_{D-2}
\end{equation}
where $S_r$ is the cross-section of $\Sigma_D$ with the sphere of radius $r$ and $d\Omega_{D-2}$ is the volume element of ($D-2$)-dimensional unit sphere.
From stationarity of the disk, we can prove the constancy of the accretion rate about both time $t$ and radius $r$.
Then we have the following expression for the integration of the continuity equation:
\begin{equation}
\llaabel{eq:jn-def}
j_n(r):=\int_{S_r}nu^r\sqrt{fg}r^{D-2}d\Omega_{D-2}=const.
\end{equation}
This is the constancy of the particle number flux $j_n(r)$.

\subsection{Construction of the accretion problem}
The system of Eqs.~(\ref{first_law}), (\ref{num_conservation}) and (\ref{em_conservation}) implies the adiabatic condition of the fluid, or equivalently, $u^\mu\partial_\mu s=0$.
The conditions for the disk, the constancy of the entropy in $\theta_i$-direction and the stationarity and the rotational symmetry of the matter distribution mean $\partial_{\theta_i}s=\partial_ts=\partial_\phi s=0$ and the adiabatic condition reduces to the condition $\partial_rs=0$.
Therefore the entropy of the fluid is constant over the disk volume and independent of the time.
Then we can write the enthalpy as a function of the number density,
\begin{equation}
\llaabel{enthalpy_expression}
h=h(n).
\end{equation}
The projection of the energy-momentum conservation equation $\nabla_\mu
T^\mu_\nu=0$ onto the direction of a Killing vector $\xi^\nu$ generally
gives conservation of the quantity $hu_\mu\xi^\mu$ along the fluid
motion if the matter distribution has the symmetry associated with the Killing vector:
\begin{equation}
u^\mu\nabla_\mu\left(hu_\nu\xi^\nu\right)=0
\end{equation}
The disk of our accretion problem has symmetries associated with the two Killing vectors, $\xi_{(t)}=\partial_t$ and
$\xi_{(\phi)}=\partial_\phi$.
Therefore we immediately get the following two integrals of motion,
\begin{eqnarray}
\llaabel{int_motion_t}
hu_t=const.,\\
\llaabel{int_motion_phi}
hu_\phi=const.,
\end{eqnarray}
corresponding to $\xi_{(t)}$ and $\xi_{(\phi)}$, respectively.
These are the specific energy and the specific angular momentum per particle and constant over the whole region in the disk due to the assumptions on the disk.
\par
From the assumption that $n$ and $u^r$ are independent of all the polar angle $\theta_i$, the particle number flux $j_n$ is calculated as,
\begin{eqnarray}
\llaabel{eq:jn-explicit}
j_n(r)
&=&2\pi\Theta\sqrt{fg}r^{D-2} nu^r,
\end{eqnarray}
where $2\pi\Theta$ is the ($D-2$)-dimensional solid angle subtended by the disk.
Together with Eq.~(\ref{int_motion_t})-(\ref{int_motion_phi}), we also have constancy of the energy flux and the
angular momentum flux,
\begin{eqnarray}
\llaabel{engy_flx}
j_\epsilon(r):=-2\pi\Theta\sqrt{fg}r^{D-2}nhu_tu^r=const.,\\
\llaabel{angl_momentum_flx}
j_\phi(r):=2\pi\Theta\sqrt{fg}r^{D-2}nhu_\phi u^r=const.
\end{eqnarray}
Thus, the integration of the conservation equations leads to the constancy of $j_n$, $j_\epsilon$ and $j_\phi$
\par
From the assumption that the disk is sufficiently thin, we estimate the values of the components of the 4-velocity through
\begin{eqnarray}
-1=g_{\mu\nu}u^\mu u^\nu
=-f\left(u^t\right)^2+g\left(u^r\right)^2+r^2\left(u^\phi\right)^2,
\end{eqnarray}
where it should be noted that $g_{\phi\phi}\to r^{2}$ in the limit of geometrically thin disk.
Introducing the fluid's angular velocity $\Omega_f:=u^\phi /u^t$, we have
\begin{equation}
\llaabel{eq:ut-exp}
\left(f-r^2\Omega_f^2\right)\left(u^t\right)^2=1+g\left(u^r\right)^2.
\end{equation}
We define the function $F$ by
\begin{equation}
F:=\left(\frac{j_\epsilon}{j_n}\right)^2
\end{equation}
and using Eq.~(\ref{eq:ut-exp}), we obtain
\begin{equation}
\llaabel{eq:F-calc}
F=h^2 u_t^2=h^2\left[f+fg\left(u^r\right)^2\right]\frac{f}{f-r^2\Omega_f^2}.
\end{equation}
The angular velocity $\Omega_f$ can be expressed as
\begin{equation}
\Omega_f=\frac{u^\phi}{u^t}=
-\frac{f}{r^2}\frac{u_\phi}{u_t}
=\omega fr^{-2}
\end{equation}
in the limit of geometrically thin disk.
The parameter $\omega:=j_\phi/j_\epsilon=-u_\phi/u_t$ is constant
due to Eqs.~(\ref{engy_flx})-(\ref{angl_momentum_flx}).
We can write $\left(u^r\right)^2$ in Eq.~(\ref{eq:F-calc}) as
\begin{equation}
\left(u^r\right)^2=\frac{\mu^2}{fgr^{2(D-2)}n^2},
\end{equation}
where Eq.~(\ref{eq:jn-explicit}) is used and the parameter $\mu:=j_n/2\pi\Theta$ is constant from Eq.~(\ref{eq:jn-def}).
Substituting the above results into Eq.~(\ref{eq:F-calc}), we finally get
\begin{equation}
F=h^2\left[f+\frac{\mu^2}{r^{2(D-2)}n^2}\right]\frac{1}{1-\omega^2fr^{-2}}.
\end{equation}
The function can be seen as a function of two variables, $r$ and $n$, and the conservation equations imply its constancy.
As a consequence, we have constructed the master equation of the accretion
problem by the following algebraic equation with two parameters, $\mu$ and $\omega$:
\begin{equation}
\llaabel{master_eq}
F(r,n)
=h^2(n)\left[f(r)+\frac{\mu^2}{r^{2(D-2)}n^2}\right]
\frac{1}{1-\omega^2f(r)r^{-2}}=const.
\end{equation}
The level curves satisfying the master equation on the phase space $(r,n)$ are the solution curves.
\par
Given $\mu$ and $\omega$, the function $F(r,n)$ is specified and
the constant in the rightmost side of Eq.~(\ref{master_eq}) determines the accretion flow, or equivalently, the solution curve on the phase space $(r,n)$.
Once the distribution of the number density $n$ is obtained, the equations $j_n=2\pi\Theta\mu$ and $j_\phi/j_\epsilon=\omega$ give the velocity distribution.
It is worth noting that the master equation coincides with that of the spherically symmetric accretion problem in paper I~\cite{koga} in the irrotational case $\omega=0$ and does not depend on the $(r,r)$-component of the metric $g(r)$.

\subsection{Critical point}
Here we give the definition of the critical point and its classification by reformulating the problem in terms of a dynamical system on the phase space $(r,n)$.
The analysis of this kind was introduced in an accretion problem by Chaverra and Sarbach~\cite{chaverra}.
Generally, the critical point plays an important role in accretion problems and is closely related to the sonic point of the flow.

\subsubsection{Definition of critical point}
In our accretion problem Eq.~(\ref{master_eq}), the solutions are
described as the level curves of the function $F(r,n)$ on
the phase space $(r,n)$.
These curves can be also obtained by integrating the ordinary differential equation,
\begin{equation}
\llaabel{master_ode}
\frac{d}{d\lambda}\left(\begin{array}{c} r\\ n \end{array}\right)
=\left(\begin{array}{r} \partial_n \\ -\partial_r \end{array}\right)F(r,n),
\end{equation}
as orbits with a parameter $\lambda$.
This is reformulation of the master equation Eq.~(\ref{master_eq}) in terms of a dynamical system with the RHS being the Hamiltonian vector field with respect to the Hamiltonian $F(r,n)$.
Then, the notion of a {\it critical point} (or stationary point as in a dynamical system) at which the RHS of Eq.~(\ref{master_ode}) vanishes arises
and its conditions are
\begin{subnumcases}
{}
\llaabel{critical-n}
\partial_nF=0\\
\llaabel{critical-r}
\partial_rF=0.
\end{subnumcases}
We define a {\it critical point} $(r_c,n_c)$ of the accretion problem as a point on the phase space $(r,n)$ at which the conditions Eqs.~(\ref{critical-n})-(\ref{critical-r}) are satisfied.

\subsubsection{Types of critical points}
\llaabel{sec:types}
The linearization of Eq.~(\ref{master_ode}) around a critical point allows us to classify the critical point into two types.
The one is a saddle point and the another one is an extremum point.
A saddle point is a point, in this case, through which two solution orbits pass.
On the other hand, orbits in vicinity of an extremum point are closed curves around the point.
\par
The linearization matrix $M_c$ is given by
\begin{equation}
M_c:=\left(
	\begin{array}{cc}
	\partial_r\partial_n & \partial_n^2 \\
	-\partial_r^2 & -\partial_r\partial_n 
	\end{array}
	\right)F(r_c,n_c).
\end{equation}
This matrix, being real, $2\times 2$ and traceless, has two eigenvalues with opposite signs.
If the determinant of the matrix,
\begin{equation}
\llaabel{eq:det-def}
\det M_c=\left(\partial_r^2F\right)_c\left(\partial_n^2F\right)_c
-\left(\partial_r\partial_nF\right)_c^2,
\end{equation}
where the subscript $c$ denotes the values at $(r_c,n_c)$, is negative (positive), the eigenvalues are real (pure imaginary).
As in a dynamical system, the real eigenvalues imply that the critical point is a saddle point.
For the imaginary eigenvalues, the orbits around the critical point are periodic in linear order.
However, because they are the contours of the real function $F(r,n)$, the orbits must be closed loops.
Therefore the imaginary eigenvalues imply an extremum point.

\subsection{Sonic point}
The sonic point  is the point on the phase space $(r,n)$ at which the
squared 3-velocity $v^2$ of the fluid equals to
the squared sound speed $v_s^2$, or~in other words, Mach number equals
to one (sometimes referred to as {\it sonic surface} because it forms a surface in the physical space).
The sonic point usually coincides with the critical point in accretion problem in a reasonable frame.
Here we introduce the fluid co-rotating frame (FCRF) and observe that
the 3-velocity $v^2$ in the FCRF
gives the coincidence between the sonic point and the critical point.

\subsubsection{Fluid co-rotating frame (FCRF)}
We refer to the fluid co-rotating frame as the frame staying at the same radius but rotating with the same angular velocity $\Omega_f$
of the fluid.
The 4-velocity $u_o$ of an observer at rest in this frame is defined as
\begin{equation}
u_o:=\gamma\left(\partial_t+\Omega_f\partial_\phi\right)
\end{equation}
where $\Omega_f:=u^\phi/u^t$ is the fluid's angular velocity.
We consider the observer on the equatorial plane.
The factor $\gamma$ is determined by the normalization of the 4-velocity, $g_{\mu\nu}u_o^\mu u_o^\nu=-1$, and we have
\begin{equation}
\gamma^{-2}=f-r^2\Omega_f^2.
\end{equation}
\par
The squared 3-velocity $v^2$ of the fluid in the FCRF is given by
\begin{eqnarray}
\frac{1}{1-v^2}=\left(g_{\mu\nu}u^\mu u_o^\nu\right)^2
=1+\frac{\mu^2}{fr^\delta n^2}
\llaabel{3-velocity}
\end{eqnarray}
or
\begin{eqnarray}
v^2=\frac{\mu^2}{\mu^2+fr^\delta n^2},
\end{eqnarray}
where $\delta:=2(D-2)$ and we have used the normalization condition of the 4-velocity, $u^\mu u_\mu = -1$.

\subsubsection{Sonic point and critical point}
\llaabel{subsub:sonic-critical}
Let us calculate explicitly the one of the conditions for the critical point, $\partial_nF(r,n)=0$,
\begin{eqnarray}
0&=&\partial_nF\nonumber\\
&=&\frac{2h^2}{n}\frac{\mu^2}{r^\delta n^2}\left(v_s^2(n)\left[1+f\frac{r^\delta n^2}{\mu^2}\right]-1\right)\frac{1}{1-\omega^2fr^{-2}}.
\end{eqnarray}
Therefore, we
can see that the sound speed $v_s^2(n):=\partial\ln h/\partial \ln n$ can be always written as
\begin{equation}
\llaabel{vs-form}
v_s^2=\frac{\mu^2}{\mu^2+fr^\delta n^2}
\end{equation}
on the point $(r,n)$ satisfying the condition $\partial_nF(r,n)=0$ including the critical point.
Conversely, if the Eq.~(\ref{vs-form}) is satisfied, the condition $\partial_nF(r,n)=0$ holds.
Comparing with Eq.~(\ref{3-velocity}), we can say that the critical point is always the sonic point.
The point satisfying $\partial_nF=0$ and $\partial_rF\neq 0$ is not the critical point but the sonic point.
However, a solution orbit $n=n_{sol}(r)$ passing such a point is typically a double-valued function around the point because $dn_{sol}/dr=-\partial_rF/\partial_nF=\pm\infty$ there and we regard such a solution as unphysical in the current paper.
Furthermore, because there does not exist a solution orbit passing extremum points, we can identify the sonic point of physically acceptable transonic flow to the critical point which is a saddle point and the following theorem holds:
\begin{theorem}\llaabel{theorem:sonic-critical}
For a physical transonic accretion flow in our accretion problem, its sonic point coincides with a critical point on
the phase space which is a saddle point.
\end{theorem}

\section{Conditions for critical point and its classification}
\llaabel{sec:criticalpoint}
We explicitly calculate the conditions for the critical point
Eqs.~(\ref{critical-n})-(\ref{critical-r}) and its classification
by the sign of the determinant of the matrix
Eq.~(\ref{eq:det-def}) without specifying the EOS of the fluid in the following.
\par
The conditions for the critical point Eqs.~(\ref{critical-n})-(\ref{critical-r}) can be explicitly written as
\begin{eqnarray}
\llaabel{cp-n}
\partial_nF
&=&\frac{2h^2}{n}\frac{\mu^2}{r^\delta n^2}\left(v_s^2\left[1+f\frac{r^\delta n^2}{\mu^2}\right]-1\right)\Omega=0,\\
\llaabel{partialrF-case2}
\partial_rF&=&h^2f\Omega\left[\frac{\left(f\Omega\right)'}{f\Omega}
-\frac{\mu^2}{fr^\delta n^2}\frac{\left(r^\delta\Omega^{-1}\right)'}{r^\delta\Omega^{-1}}\right]=0,
\end{eqnarray}
where the function $\Omega(r)$ is defined by
\begin{equation}
\Omega(r):=\frac{1}{1-\omega^2 fr^{-2}}.
\end{equation}
It should be noted that $\Omega(r)>0$ is always satisfied because of Eq.~(\ref{master_eq}).
\par
We can prove that $\left(r^\delta\Omega^{-1}\right)'\ne0$ and $(f\Omega)'\ne0$ at the critical point (See Appendix.\ref{app:nonzero}) and therefore Eq.~(\ref{partialrF-case2}) is solved for $n^2$ with the help of $\left(r^\delta\Omega^{-1}\right)'\ne0$ and $(f\Omega)'\ne0$.
Then we get the condition for the radius $r_c$ of critical point eliminating the number density $n$ from Eqs.~(\ref{cp-n})-(\ref{partialrF-case2}) as follows.
\begin{eqnarray}
\llaabel{eq:cradius-condition}
\mathcal{F}(r)
&:=&v_s^2\left(\bar{n}(r)\right)\left[1+\frac{\left(r^{\delta}\Omega^{-1}\right)'}{r^{\delta}\Omega^{-1}}\frac{f\Omega}{\left(f\Omega\right)'}\right]-1=0,\\
\bar{n}(r)
&:=&\frac{\left|\mu\right|}{r^{\delta}\Omega^{-1}}\sqrt{\frac{\left(r^{\delta}\Omega^{-1}\right)'}{\left(f\Omega\right)'}}
\end{eqnarray}
Once the radius $r_c$ is obtained, the number density $n_c$ of the critical point is uniquely determined by
\begin{equation}
\llaabel{eq:cnumber-condition}
n_c=\bar{n}(r_c).
\end{equation}
\par
The determinant of the matrix $M_c$ can be expressed in terms of the function $\mathcal{F}(r)$ defined above:
\begin{equation}
\det M_c=-\frac{\left(r^{\delta}\Omega^{-1}\right)_c}{\left(r^{\delta}\Omega^{-1}\right)'_c}\frac{4h_c^4}{n_c^2}\left((f\Omega)'_c\right)^2\mathcal{F}'(r_c)
\end{equation}
Therefore, using the fact that $(r^\delta \Omega^{-1})'_c$ and $(f\Omega)'_c$ have the same sign, we can uniquely classify the critical point through the sign of the derivative of the function $\mathcal{F}'$ and the factor $(f\Omega)'$ at the critical radius:
\begin{equation}
\llaabel{rot-classification}
saddle\ (extremum)\ point \Leftrightarrow \left(f\Omega\right)'_c{\mathcal{F}}'(r_c)>0\ (<0)
\end{equation}

\section{Correspondence between sonic point and photon sphere for photon gas accretion}
\llaabel{sec:radiation}
In this section, we analyze the critical point in the case of ideal photon gas
and, as the main result of the current paper, prove the following theorem about the correspondence between the sonic points and the photon spheres.
\begin{theorem}\llaabel{theorem:correspondence}
For a physical transonic accretion flow of ideal photon gas fluid of our accretion problem Eq.~(\ref{master_eq}), its sonic point(s) must be on (one of) the unstable photon sphere(s) of the spacetime Eq.~(\ref{eq:metric}).
\end{theorem}
\par
We derived the EOS of ideal photon gas in arbitrary spatial dimensions $d$ in paper I~\cite{koga}.
For our purpose, it is sufficient to know that the enthalpy can be written in the form,
\begin{equation}
h(n)=(const.)\times n^{\gamma-1},
\end{equation}
where the index $\gamma$ is related to the dimension by $\gamma=(d+1)/d$.
The sound speed $v_s^2(n)$ is then computed as
\begin{equation}
\llaabel{vs-radiation}
v_s^{2}(n):=\frac{\partial \ln h}{\partial \ln n}=\gamma-1=\frac{1}{d}=\frac{1}{D-1}.
\end{equation}

\subsection{Critical point}
For the conditions for the critical point $(r_c,n_c)$ and its classification for the ideal photon gas accretion, we have the following lemma.
\begin{lemma}\llaabel{lemma:photongascritical}
For the accretion of ideal photon gas in our accretion problem, radius $r_c$ of a critical point
is specified by
\begin{equation}
\llaabel{cradius-radiation}
(fr^{-2})'=0
\end{equation}
and the corresponding critical density $n_c$ is
\begin{equation}
n_c=\left|\mu\right|\sqrt{\frac{\delta}{r_c^{\delta+1}f'_c}}.
\end{equation}
The type of the critical point is classified by the equation
\begin{equation}
\llaabel{rot-classification-radiation}
saddle\ (extremum)\ point \Leftrightarrow \left(fr^{-2}\right)''_{r=r_c}<0\ (>0).
\end{equation}
\end{lemma}
\begin{proof}
The critical radius $r_c$ is determined by Eq.~(\ref{eq:cradius-condition}).
Substituting the sound speed of radiation fluid into Eq.~(\ref{vs-radiation}), the condition for the critical radius $r_c$ is given by
\begin{equation}
\mathcal{F}=-\frac{1}{\delta+2}\frac{f\Omega}{(f\Omega)'}\frac{\delta+2\omega^2fr^{-2}}{fr^{-2}\left(1-\omega^2fr^{-2}\right)}\left(fr^{-2}\right)'=0.
\end{equation}
As mentioned above Eq.~(\ref{eq:cradius-condition}), $(f\Omega)'\ne0$ is always satisfied at critical point.
Therefore, the critical radius is obtained by solving the equation
\begin{equation}
\left(fr^{-2}\right)'=0.
\end{equation}
Once the radius $r_c$ is obtained, we get the corresponding number density $n_c$ from Eq.~(\ref{eq:cnumber-condition}),
\begin{equation}
n_c=\left|\mu\right|\sqrt{\frac{\delta}{r_c^{\delta+1}f'_c}},
\end{equation}
where we used the facts $\Omega'(r_c)=0$ and $f_c'=2f_c/r_c$.
This expression is also independent of the parameter $\omega$.
\par
From Eq.~(\ref{rot-classification}) and $(f\Omega)'_c=(f'\Omega)_c>0$, the types of the critical point $(r_c,n_c)$ are determined by the sign of $\mathcal{F}'(r_c)$.
Using the fact that $\Omega'(r_c)=0$ and
$f_c'=2f_c/r_c$, we can explicitly write
${\mathcal{F}}'(r_c)$ as
\begin{equation}
\mathcal{F}'(r_c)=-\frac{1}{\delta+2}\left[\frac{r}{2}\frac{\delta+2\omega^2fr^{-2}}{fr^{-2}\left(1-\omega^2fr^{-2}\right)}\left(fr^{-2}\right)''\right]_{r=r_c}.
\end{equation}
Then we get the explicit form of Eq.~(\ref{rot-classification}) for radiation fluid accretion:
\begin{equation}
saddle\ (extremum)\ point \Leftrightarrow \left(fr^{-2}\right)''_{r=r_c}<0\ (>0)
\end{equation}
\end{proof}
\par
It should be noted that the conditions do not depend on the parameter of
the rotation $\omega$ and thereby coincides with the condition in the
case of spherical flows.

\subsection{Proof of Theorem: Correspondence among the points}
\llaabel{sec:correspondence}
In the following, we see the correspondence among the three objects; the photon sphere, the critical point and the sonic point of our accretion problem with fluid of ideal photon gas and finally prove Theorem~\ref{theorem:correspondence}.
\par
In paper I~\cite{koga}, we derived the following lemma about the conditions for photon spheres of the spacetime Eq.~(\ref{eq:metric}):
\begin{lemma}\llaabel{lemma:photonsphere}
The radius photon sphere is specified by the equation,
\begin{equation}
\llaabel{eq:pp-radius}
(fr^{-2})'=0.
\end{equation}
The stability condition of the photon sphere is given by
\begin{equation}
\llaabel{eq:pp-stability}
stable\ (unstable) \Leftrightarrow (fr^{-2})''>0 \ (<0)
\end{equation}
at the radius of the photon sphere.
\end{lemma}
\par
The conditions of the critical radius Eq.~(\ref{cradius-radiation}) and its classification Eq.~(\ref{rot-classification-radiation}) coincide with that of the photon sphere Eqs.~(\ref{eq:pp-radius}) and~(\ref{eq:pp-stability}), respectively.
Then we immediately obtain the following corollary about the correspondence between photon spheres and critical points of ideal photon gas accretion in our accretion problem from Lemma~\ref{lemma:photongascritical} and~\ref{lemma:photonsphere}:
\par
\begin{corollary}
If the spacetime has photon spheres, there exists a critical point at the same radius for each of the spheres.
Furthermore, for an unstable photon sphere, the critical point at the same radius is always a saddle point while for a stable one, the corresponding critical point is an extremum point.
\end{corollary}
\par
Note that if a critical point exists, there must be a photon sphere at the same radius and the extremum(saddle) point corresponds to the stable(unstable) sphere.
There is a one-to-one correspondence between critical points and photon spheres.
It is worth noting that if the spacetime has more than one photon spheres, the stable and unstable spheres appear alternately as we can see from Eqs.~(\ref{eq:pp-radius})-(\ref{eq:pp-stability}).
The fact also leads to the alternate appearance of the corresponding extremum and saddle points on the phase space $(r,n)$.
\par
As mentioned in Theorem~\ref{theorem:sonic-critical} in Sec.~\ref{subsub:sonic-critical}, the sonic point of the physical transonic accretion flow coincides with a critical and saddle point.
Then Theorem~\ref{theorem:sonic-critical} and Corollary above immediately prove Theorem~\ref{theorem:correspondence}.
\par
Even if the accretion fluid flow is rotational, there exists a correspondence between the photon sphere and the sonic point of the radiation fluid accretion as far as the fluid rotates around the center on the equatorial plane and satisfies our disk conditions.

\section{Conclusion}
\llaabel{sec:conclusion}
We formulated the rotational accretion problem of the disk lying on the equatorial plane of 
the $D$-dimensional static, spherically symmetric spacetime.
We adopted the simplest accretion disk model similar to the one given by Abraham et al.~\cite{abraham}.
After giving the definition of a critical point and observing its relation to the sonic point of a transonic accretion flow, we showed the explicit form of the conditions of the critical point and its classification for the perfect fluid with arbitrary EOS.
\par
Applying the EOS of radiation fluid to the analysis, we proved the existence of the correspondence between the sonic points of rotational ideal photon gas accretion and the photon spheres of the spacetime.
We showed, at first, that a critical point that is a saddle point is
always on the unstable photon sphere while a critical point that
is an extremum point is always on the stable photon sphere.
Then, from the relation between the critical point
and the sonic point, we proved a correspondence between the photon sphere and the sonic point of radiation fluid in the rotational case.
Then we proved that the physical transonic accretion flow must have its
sonic point on (one of) the unstable photon sphere(s) of the spacetime.
The result holds in arbitrary dimensions of the spacetime as in
the case of spherical flows shown in paper I~\cite{koga}.
\par
We proved the correspondence in our idealized disk model, where
the critical and the sonic points are directly related.
However, there are many other possibilities about disk configurations such as disk thickness and vertical equilibrium and it is known that the Mach number can deviate from one on the critical point depending on the disk models~\cite{das}~\cite{barai}.
Therefore it is not so clear whether the correspondence exists between the sonic point and photon sphere, or, between the critical point and photon sphere in more realistic models, e.g., in the model by Abramowicz et al.~\cite{abramowicz}.
\par
If we focus on a spherical star in vacuum with general relativity, its exterior geometry is described by the Schwarzschild spacetime.
As considered in astrophysical cases, such a star typically has a radius
greater than $3M$ for the mass $M$ and does not have a photon sphere.
Our result says that the star of radius greater than $3M$
cannot have transonic accretion of radiation fluid as far as 
it meets the conditions in the
current paper or paper I~\cite{koga},
that is, rotating geometrically thin accretion disks or 
spherically symmetric accretion flows.
\par
We have extended the correspondence from 
spherically symmetric flows~\cite{koga} to rotating geometrically 
thin accretion disks.
This result, more and more, motivates us to consider its physical
origin, that is to say, the connection between the behavior of free photons on the geometry and 
behavior of photon gas.
This is also attractive as one of problems concerning hydrodynamics on curved spacetime.
\par
It is also interesting to investigate a correspondence in a spacetime with a rotating object.
The central object in such a spacetime has angular momentum and there is
preferred direction of rotation for the motion of matter.
The radius of the circular orbits of particles depends on the direction of rotation and also that of the sonic point does.
The physics becomes more complicated but more interesting 
in that situation.

\begin{acknowledgments}
The authors thank T. Igata and K. Ogasawara for their very helpful discussions.
This work was partially supported by JSPS KAKENHI Grant No. JP26400282
(T.H.).
\end{acknowledgments}

\appendix
\section{Nonzero terms on the critical point}
\llaabel{app:nonzero}
Here we show that $\left(r^\delta\Omega^{-1}\right)'\ne0$ and $(f\Omega)'\ne0$ at the critical point.
Consider the case in which the two conditions, Eq.~(\ref{partialrF-case2}) and
\begin{equation}
\llaabel{assumption-case2}
\left(r^\delta\Omega^{-1}\right)'=0,
\end{equation}
are satisfied, simultaneously.
Clearly, this is equivalent to the conditions,
\begin{subnumcases}
{}
\llaabel{singular-condition-1}
\left(r^\delta\Omega^{-1}\right)'=0\\
\left(f\Omega\right)'=0.
\end{subnumcases}
Eliminating $\Omega$ from the above two equations, we get
\begin{equation}
\llaabel{cradius-case2}
(fr^{\delta})'=0.
\end{equation}
This condition solely determines the critical radius $r_c$.
Once the radius $r_c$ is obtained, we can specify the parameter $\omega^2$ from the above equations.
From Eq.~(\ref{cradius-case2}), we find the expression,
\begin{equation}
(fr^{-2})'=(fr^{\delta-(2+\delta)})'=fr^{\delta}(r^{-(2+\delta)})'=-(2+\delta)fr^{-3},
\end{equation}
for $r=r_c$.
Substituting the result into Eq.~(\ref{singular-condition-1}), we get, at $r=r_c$,
\begin{eqnarray}
0
=\left(r^\delta\Omega^{-1}\right)'
=r^{\delta-1}\left(\delta+2\omega^2fr^{-2}\right).
\end{eqnarray}
By definition, both the terms in the bracket must be positive and we do not have the parameter $\omega$ satisfying the condition.
Therefore the conditions for the critical point are always accompanied by the additional conditions $\left(r^\delta\Omega^{-1}\right)'\ne0$ and $(f\Omega)'\ne0$.
\par
Note that the above proof relies on the only one of the critical conditions Eq.~(\ref{partialrF-case2}), $\partial_rF(r,n)=0$.
This implies that the conditions $\left(r^\delta\Omega^{-1}\right)'\ne0$ and $(f\Omega)'\ne0$ hold anywhere on the curve $\Gamma$ defined by the equation $\partial_rF(r,n)=0$ and the curve $\Gamma$ can be expressed by $n=\bar{n}(r)$.

%
\bibliography{cp_pp_correspondence_rot}

\end{document}